\documentclass[journal,twoside,web]{ieeecolor}
\pdfoutput=1
\usepackage{generic}
\usepackage{cite}
\usepackage{amsmath,amssymb,amsfonts}
\usepackage{algorithm}
\usepackage[justification=centering]{caption}
\usepackage[noend]{algpseudocode}
\usepackage{algorithmicx}
\usepackage{graphicx}
\usepackage{textcomp}
\usepackage{bm}
\usepackage{color}
\usepackage{mathrsfs}

\newtheorem{remark}{Remark}[section]
\newtheorem{theorem}{Theorem}[section]
\newtheorem{assumption}{Assumption}[section]
\newtheorem{lemma}{Lemma}[section]

\newtheorem{definition}{Definition}[section]

\def\ban{\begin{eqnarray*}}
	\def\ean{\end{eqnarray*}}
\def\bna{\begin{eqnarray}}
	\def\ena{\end{eqnarray}}
\def\BibTeX{{\rm B\kern-.05em{\sc i\kern-.025em b}\kern-.08em
    T\kern-.1667em\lower.7ex\hbox{E}\kern-.125emX}}
\begin{document}
\title{Stability of  FFLS-based Diffusion Adaptive Filter  Under  Cooperative Excitation Condition}
\author{Die Gan,  Siyu Xie, Zhixin Liu, \IEEEmembership{Member, IEEE}, and Jinhu L{\"{u}}, \IEEEmembership{Fellow, IEEE}
\thanks{Corresponding author: Zhixin Liu.}
\thanks{This work was supported by Natural Science Foundation of China under Grant T2293772, the National Key R\&D Program of China under Grant 2018YFA0703800, the Strategic Priority Research Program of Chinese Academy of Sciences under Grant No. XDA27000000,  and National Science Foundation of Shandong Province (ZR2020ZD26).}
\thanks{D. Gan is with the Zhongguancun  Laboratory, Beijing, China (e-mail: gandie@amss.ac.cn).}
\thanks{S. Y. Xie is with the School of Aeronautics and Astronautics, University of Electronic Science and Technology of China, Chengdu 611731, China (e-mail: syxie@uestc.edu.cn).}
\thanks{Z. X. Liu is  with the  Key Laboratory of Systems and Control, Academy of Mathematics and Systems Science, Chinese Academy of Sciences,  and School of Mathematical Sciences, University of Chinese Academy of Sciences, Beijing,  China. (e-mail: lzx@amss.ac.cn.).}
\thanks{J. H. L{\"{u}  is with the School of Automation Science and Electrical Engineering,  Beihang	University, Beijing, China, and also with the Zhongguancun  Laboratory, Beijing, China (e-mail: jhlu@iss.ac.cn)}.}
}
\maketitle

\begin{abstract}
 In this paper, we consider the distributed filtering problem over sensor networks such that all sensors cooperatively track unknown time-varying parameters by using local
 information. A distributed forgetting factor least squares (FFLS) algorithm is proposed by minimizing a local cost function formulated as a linear combination of accumulative estimation error. Stability analysis of the algorithm is provided under a cooperative excitation condition which contains spatial union information to reflect the cooperative effect of all sensors. Furthermore, we generalize theoretical results to the case of Markovian switching directed graphs. The main difficulties of theoretical analysis lie in how to analyze properties of the product of non-independent and non-stationary random matrices. Some techniques such as stability theory, algebraic graph theory and Markov chain theory are employed to deal with the above issue. Our theoretical results are obtained without relying on the independency or stationarity assumptions of regression vectors which are commonly used in existing literature. 
 \end{abstract}


\section{Introduction}\label{sec:introduction}
\IEEEPARstart{O}{wing} to the capability to process the collaborative data, wireless sensor networks (WSNs) have attracted increasing research attention
in diverse areas, including  consensus seeking \cite{Ren2005}\cite{Wangyupeng2015}, resource allocation \cite{Kaihong2022}\cite{wangbo2021}, and formation control \cite{Linzhiyu2014}\cite{YongranZhi2022}.
How to design the distributed adaptive estimation and filtering algorithms to cooperatively estimate unknown parameters has become one of the most important research topics.
Compared with centralized estimation algorithms where a fusion center is needed to collect and process   information measured by all  sensors,
the distributed ones can estimate or track an unknown parameter process of interest cooperatively  by using local noisy measurements.
Therefore,
the distributed algorithms are easier to be implemented because of their robustness to network link failure,  privacy protection, and  reduction on communication and computation costs.


Based on classical estimation algorithms and typical distributed strategies such as the incremental,  diffusion and consensus,
a number of distributed adaptive estimation or filtering algorithms  have been investigated (cf., \cite{Battistelli2014,Chenweisheng2014,Javed2022,distributedsg2,Schizas2009,Zhangling2017,Takahashi2010,Lei2015,Mateos2012,Xie20181,Gan2022,Gan202260,Gan2021}), e.g., the consensus-based least mean squares (LMS), the diffusion Kalman filter (KF), the diffusion least squares (LS), the incremental LMS, the combination of diffusion and consensus stochastic gradient (SG), the diffusion forgetting factor least squares (FFLS). Performance analysis of the distributed  algorithms is also studied under some information conditions. For
deterministic signals or deterministic system matrices,  Battistelli and Chisci in \cite{Battistelli2014} provided the  mean-square boundedness of the state estimation error of the distributed Kalman filter algorithm under a collectively observable condition.
Chen et al. in \cite{Chenweisheng2014} studied the convergence of distributed adaptive identification algorithm under a cooperative persistent excitation (PE) condition. Javed et al. in \cite{Javed2022} presented stability analysis of the cooperative gradient algorithm for the deterministic regression vectors satisfying a cooperative PE condition.
Note that the signals are often random since they are generated from dynamic systems affected by noises. For the random regression vector case, Barani et al. in \cite{distributedsg2} studied the convergence of distributed stochastic gradient descent algorithm with independent and identically distributed (i.i.d.) signals.
Schizas et al.  in \cite{Schizas2009}  provided the stability analysis of a distributed LMS-type adaptive algorithm under the strictly stationary and ergodic regression vectors.
Zhang et al. in \cite{Zhangling2017} studied the  mean
square performance of a diffusion FFLS algorithm with independent input signals.
Takahashi et al. in \cite{Takahashi2010} established the performance analysis of the diffusion LMS algorithm for i.i.d.  regression vectors.  Lei and Chen in \cite{Lei2015} established the convergence analysis of the distributed stochastic approximation algorithm with ergodic system signals.
Mateos and Giannakis  in \cite{Mateos2012} presented the stability and performance analysis of the distributed FFLS algorithm under
the spatio-temporally white regression vectors condition.

We remark that most theoretical results mentioned in the above literature were established by requiring regression vectors to be either
deterministic and satisfy PE conditions, or random but satisfy independency, stationarity and ergodicity conditions. In fact, the observed data
are often random and  hard to satisfy the above  statistical assumptions,  since they are generated by complex dynamic systems  where feedback loops inevitably exist (cf., \cite{Guo2020509}).
The main  difficulty in  performance analysis of distributed algorithms is to analyze the product of random matrices involved in estimation error equations.
In order to relax the above stringent conditions on random regression vectors, some
progress has been made on distributed adaptive estimation and filtering algorithms under undirected graphs.
For estimating time-invariant parameters,
the convergence analysis of distributed SG algorithm and distributed LS algorithm is provided in \cite{Gan2019} and \cite{Xie2021} under  cooperative excitation conditions.
For tracking a time-varying parameter,
Xie and Guo in \cite{Xie20181} and \cite{Xie20182} proposed the weakest possible cooperative information conditions to guarantee the stability and performance of consensus-based and diffusion-based LMS algorithms.
Compared with LMS algorithm, FFLS algorithm can generate more accurate estimates  in the transient phase (see e.g.,\cite{Macchi1988}), and the stability analysis for the distributed FFLS algorithm is still lacking. In this paper, we focus on the design and stability analysis of distributed FFLS algorithm without relying on the independency, stationarity or ergodicity assumptions on regression vectors.


The information exchange between sensors is an important factor for the performance of distributed estimation algorithms, and previous studies often assume that the networks are undirected and time-invariant. In practice, they might not be bidirectional or time-invariant due to the heterogeneity of sensors and signal losses caused by the temporary deterioration in the communication link.
One approach is to model the networks which randomly change over time as an i.i.d. process, see e.g., \cite{Hatano2005,Kar2012}. However, the loss of connection usually occurs with correlations \cite{Matei2008}. Another approach is to model the random switching process as a Markov chain whose
states correspond to possible communication topologies, see \cite{Matei2008,Keyou2013,Wangyunpeng2015,Meng2018} among many others. Some studies on the distributed algorithms with deterministic or temporally independent measurement matrix under  Markovian switching topologies
are given in e.g.,\cite{Zhang2012,Liu2019}.

In this paper, we consider the distributed filtering problem  over sensor networks where all sensors aim at collectively tracking an unknown randomly time-varying parameter vector. Based on the fact that  recent observation data respond to the parameter changes faster than the early data,
	we introduce a forgetting factor into the local accumulative cost function  formulated as a linear combination of local estimation errors between the observation signals and the prediction signals. By minimizing the local cost function, we propose the distributed FFLS algorithm based on the diffusion strategy over the fixed undirected graph. 
	The stability analysis of the distributed FFLS algorithm is provided under a 
	cooperative excitation condition. Moreover,
	we generalize the  theoretical results to the case of Markovian switching directed sensor networks. The key difference from the fixed undirected graph case is that the adjacency matrix is an asymmetric random matrix. We employ the Markov chain theory to deal with the coupled relationship between random adjacency matrices and random regression vectors.
	The main contributions of this paper can be summarized as the following aspects:

\begin{itemize}
	
	\item  In comparison with \cite{Xie20181} and  \cite{Gan2019}, the main difficulty is that the random matrices in the error equation of the diffusion FFLS algorithm are not symmetric and the adaptive gain is no longer a scalar. We establish the exponential stability of the homogeneous part of the estimation error equation and the bound of the tracking error by virtue of the specific structure of the proposed diffusion FFLS algorithm and  stability theory of stochastic dynamic systems. 
	\item Different from the theoretical results of distributed FFLS algorithms in  \cite{Zhangling2017} and \cite{Mateos2012} where regression vectors are required to satisfy the independent or spatio-temporally uncorrelated assumptions, our theoretical analysis is obtained without relying on such stringent conditions, which makes it possible to be applied  to the stochastic feedback systems. 
	\item The cooperative excitation condition introduced in this paper is a temporal and spatial union information condition on the random regression vectors, which can reveal the cooperative effect of multiple sensors in a certain sense, i.e., the whole sensor network can cooperatively finish the estimation task, even if any individual sensor cannot due to lack of necessary information.
	
\end{itemize}

The remainder of this paper is organized as follows. In
Section \ref{sec:formulation}, we give the problem formulation of this paper.
Section \ref{distributedFFLS}  presents the distributed FFLS algorithm.
The stability of the proposed algorithm under fixed undirected graph and Markovian switching directed graphs are given in Section \ref{stability_undirected} and Section \ref{stability_Markovian},  respectively.
Finally, we
conclude the paper with some remarks in Section \ref{section_conclusion}.

\section{Problem Formulation}\label{sec:formulation}
\subsection{Matrix theory}
In this paper, we use $\mathbb{R}^m$ to denote the set of $m$-dimensional real vectors, $\mathbb{R}^{m\times n}$ to denote the set of real matrices with $m$ rows and $n$ columns, and $\bm I_m$ to denote the $m$-dimensional square identity matrix.
For a matrix $\bm A\in\mathbb{R}^{m\times n}$, $\|\bm A\|$ denotes its Euclidean norm, i.e., $\|\bm A\|\triangleq(\lambda_{\max}(\bm A\bm A^T))^{\frac{1}{2}}$, where the notation $T$ denotes the transpose operator and $\lambda_{\max}(\cdot)$ denotes the largest eigenvalue of the matrix. Correspondingly,  $\lambda_{\min}(\cdot)$ and $tr(\cdot)$  denote the smallest eigenvalue and the trace of the matrix, respectively. The notation ${\rm {col}}(\cdot,\cdots,\cdot)$ is used to denote a vector stacked by the specified vectors, and ${\rm{diag}}(\cdot,\cdots,\cdot)$ is used to denote a block matrix formed in a diagonal manner of the corresponding vectors or matrices.

For a matrix $\bm A=[a_{ij}]\in \mathbb{R}^{m\times m}$, if $\sum_{j=1}^m a_{ij}=1$ holds for all $i=1,\cdots, m$, then it is called stochastic.
The Kronecker product of two matrices $\bm A$ and $\bm B$ is denoted by $\bm A\otimes \bm B$.
For two real symmetric matrices $\bm X\in\mathbb{R}^{n\times n}$ and $\bm Y\in\mathbb{R}^{n\times n}$, $\bm X\geq\bm Y$ ($\bm X>\bm Y$,  $\bm X\leq\bm Y$, $\bm X<\bm Y$) means that $\bm X-\bm Y$ is a semi-positive (positive, semi-negative, negative) definite matrix.
For a matrix sequence $\{\bm A_t\}$  and a positive scalar sequence $\{a_t\}$,  the equation $\bm A_t = O(a_t)$ means that there exists a positive constant $C$  independent of $t$ and $a_t$ such that
$\| \bm A_t\|  \leq  C a_t$ holds for all $t \geq 0$.

The  matrix inversion formula is often used in this paper and we list it as follows.
\begin{lemma}[Matrix inversion formula \cite{Zielke1968}] \label{wl1}
	For any matrices $\bm A$, $\bm B$,  $\bm C$ and  $\bm D$ with suitable dimensions, the following formula
	$$(\bm A+\bm B\bm D\bm C)^{-1}=\bm A^{-1}-\bm A^{-1}\bm B(\bm D^{-1}+\bm C\bm A^{-1}\bm B)^{-1}\bm C \bm A^{-1}	$$
	holds, provided that the relevant matrices are invertible.
\end{lemma}

\subsection{Graph theory}
We use graphs to model the communication topology between sensors.
A directed graph $\mathcal{G}=(\mathcal{V}, \mathcal{E}, \mathcal{A})$ is composed of a vertex set $\mathcal{V}=\{1,2,3,\cdots, n\}$ which stands for the set of sensors (i.e., nodes),  $\mathcal{E}\subset \mathcal{V} \times\mathcal{V}$ is the edge set, and $\mathcal{A}=[a_{ij}]_{1\leq i,j\leq n}$ is the weighted adjacency matrix. A directed edge $(i, j)\in \mathcal{E}$ means that the $j$-th sensor can receive the data from the $i$-th sensor, and sensors
$i$ and $j$ are called the parent and child sensors, respectively.
 The elements of matrix $\mathcal{A}$ satisfy $a_{ij}>0$ if $(i, j)\in \mathcal{E}$ and $a_{ij}=0$ otherwise.
 The in-degree and out-degree of sensor $i$ are defined by $\deg_{in}(i)=\sum^n_{j=1}a_{ji}$ and $\deg_{out}(i)=\sum^n_{j=1}a_{ij}$
respectively. The digraph  $\mathcal{G}$
 is called balanced if $\deg_{in}(i)=\deg_{out}(i)$ for $i=1,...,n$.
 Here, we assume that $\mathcal{A}$  is a stochastic matrix.
The neighbor set of $i$ is denoted as $\mathcal{N}_i=\{j\in \mathcal{V}, (j,i)\in\mathcal{E} \}$, and the sensor $i$ is also included in this set.
For a given positive integer $k$, the union of $k$ digraphs $\{\mathcal{G}_j=(\mathcal{V}, \mathcal{E}_j,\mathcal{A}_j), 1\leq j\leq k\}$ with the same node set is denoted by $\cup^{k}_{j=1}\mathcal{G}_j=(\mathcal{V}, \cup^k_{j=1}\mathcal{E}_j,\frac{1}{k}\sum^k_{j=1}\mathcal{A}_j)$.
A directed path from
$i_1$ to $i_l$ consists of a sequence of sensors $i_1, i_2,...i_l (l \geq 2)$, such that $(i_k,i_{k+1})\in\mathcal{E}$ for $k=1,...,l-1$.
The digraph $\mathcal{G}$ is said to be
strongly connected if for any senor there exist directed paths from this sensor
to all other sensors.
 For the graph $\mathcal{G}=(\mathcal{V}, \mathcal{E},\mathcal{A})$,
if $a_{ij}=a_{ji}$ for all $i,j\in\mathcal{V}$, then it is called an  undirected graph. The diameter $D_{\mathcal{G}}$ of the undirected graph $\mathcal{G}$ is defined as the maximum shortest length of paths between any two sensors.

\subsection{Observation model}
Consider a network consisting of $n$ sensors (labeled $1,\cdots,n$) whose task is to estimate an unknown time-varying parameter $\bm\theta_t$ by cooperating with each other. We assume that the measurement $\{y_{t,i}, \bm\varphi_{t,i}\}$  at the sensor $i$ obeys the following discrete-time stochastic regression model,
\bna
y_{t+1,i}=\bm\varphi_{t,i}^T\bm \theta_t+w_{t+1,i},\label{model}
\ena
where $y_{t,i}$ is the scalar output of the sensor $i$ at time $t$, $\bm\varphi_{t,i}\in\mathbb{R}^m$ is the random regression vector,
	$\{w_{t,i}\}$ is a noise process, and $\bm\theta_t$ is the unknown $m$-dimensional time-varying parameter whose variation at time $t$ is denoted by $\Delta\bm\theta_t$, i.e.,
	\bna
	\Delta\bm\theta_t\triangleq \bm\theta_{t+1}-\bm\theta_{t}, ~~t\geq0.\label{de1}
	\ena
Note  that when $\Delta\bm\theta_t\equiv0$, $\bm\theta_t$
becomes a constant vector.
For the special case where $w_{t+1,i}$ is a moving average process and $\bm\varphi_{t,i}$ consists of current
and past input-output data, i.e., 
\ban
\bm\varphi_{t,i}^T=[y_{t,i},\cdots,y_{t-p,i}, u_{t,i},\cdots,u_{t-q,i}]
\ean
 with $u_{t,i}$ being the  input signal of the sensor $i$ at  time $t$, then the
model (\ref{model}) can be reduced to ARMAX model with time-varying coefficients. 

\section{The distributed FFLS Algorithm}\label{distributedFFLS}
Tracking a time-varying signal  is
a fundamental problem in system identification and signal processing. The well-known recursive least squares estimator with a constant  forgetting factor $\alpha\in(0,1)$ is often used to track time-varying parameters, which
is defined by
 \bna
 \bm{\hat\theta}_{t+1,i}\triangleq \arg\min_{\bm\beta}\sum^{t}_{k=0}\alpha^{t-k}
(y_{k+1,i}-{\bm\beta}^T\bm\varphi_{k,i})^2.\label{FFLSestimator}
\ena
With some simple manipulations using the matrix inversion formula, we can obtain the following recursive FFLS algorithm (Algorithm \ref{algorithm1}) for an individual sensor.
\begin{algorithm}[ht]
	\caption{Standard non-cooperative FFLS algorithm}\label{algorithm1}
	For any given sensor $i\in\{1,...,n\}$, begin with an initial estimate $\bm {\hat \theta}_{0,i}\in\mathbb{R}^m$ and an initial positive definite matrix $\bm {P}_{0,i}\in\mathbb{R}^{m\times m}$. The standard FFLS is recursively defined at time $t\geq 0$ as follows,\\
	\begin{align*}
		\bm{\hat\theta}_{t+1,i}&=\bm{\hat\theta}_{t,i}+\frac{\bm P_{t,i}\bm\varphi_{t,i}}{\alpha+\bm\varphi^T_{t,i}\bm P_{t,i}\bm\varphi_{t,i}}(y_{t+1,i}-\bm\varphi^T_{t,i}\bm{\hat\theta}_{t,i}),\\
		\bm{ P}_{t+1,i}&=\frac{1}{\alpha}\left(\bm P_{t,i}-\frac{\bm P_{t,i}\bm\varphi_{t,i}\bm\varphi^T_{t,i}\bm P_{t,i}}{\alpha+\bm\varphi^T_{t,i}\bm P_{t,i}\bm\varphi_{t,i}}\right).
	\end{align*}
\end{algorithm}

However, due to the limited sensing ability of each sensor, it is often the case where the measurements obtained by each sensor can only reflect partial information of the unknown parameter. In such a case, if only local measurements of the sensor itself
are utilized to perform the estimation task (see Algorithm \ref{algorithm1}), then at most  part of the unknown parameter rather than the whole vector can be estimated.
Thus, in this paper, we aim at
designing a distributed adaptive estimation algorithm such that all sensors cooperatively track the
unknown time-varying parameter $\bm\theta_t$ by using random regression vectors and the observation signals from its neighbors. To simplify the analysis, in this section, we use a fixed undirected  graph $\mathcal{G}=(\mathcal{V},\mathcal{E},\mathcal{A})$ to model the communication topology of $n$ sensors.

We first introduce the following  local cost function  $\sigma_{t+1,i}(\bm\beta)$ for each sensor $i$ at the time instant $t\geq 0$ recursively formulated as
a linear combination of its neighbors' local estimation error between the observation signal and the prediction signal,
\begin{equation}\label{least}
		\sigma_{t+1,i}(\bm\beta)=\sum_{j\in\mathcal{N}_i}a_{ij}\bigg(\alpha\sigma_{t,j}(\bm\beta)
		+(y_{t+1,j}-{\bm\beta}^T\bm\varphi_{t,j})^2\bigg).
\end{equation}
with $ \sigma_{0,i}(\bm\beta)=0$.
Set
\begin{align*}
&\bm{\sigma}_{t}(\bm\beta)={\rm{col}}\{\sigma_{t,1}(\bm\beta),\cdots,\sigma_{t,n}(\bm\beta)\},\\
&\bm e_{t+1}(\bm\beta)={\rm{col}}\{(y_{t+1,1}-{\bm\beta}^T\bm\varphi_{t,1})^2,
\cdots,(y_{t+1,n}-{\bm\beta}^T\bm\varphi_{t,n})^2\}.
\end{align*}
Hence by (\ref{least}), we have
\ban
\bm{\sigma}_{t+1}(\bm\beta)&=&\alpha\mathcal{A}\bm{\sigma}_{t}(\bm\beta)+\mathcal{A}\bm e_{t+1}(\bm\beta)\\
&=&\alpha^2\mathcal{A}^2\bm{\sigma}_{t-1}(\bm\beta)+\alpha\mathcal{A}^2\bm e_{t}(\bm\beta)+\mathcal{A}\bm e_{t+1}(\bm\beta)\\
&=&\cdots\\
&=&\alpha^{t+1}\mathcal{A}^{t+1}\bm{\sigma}_{0}(\bm\beta)+\sum^t_{k=0}\alpha^{t-k}\mathcal{A}^{t+1-k}\bm e_{k+1}(\bm\beta)\\
&=&\sum^t_{k=0}\alpha^{t-k}\mathcal{A}^{t+1-k}\bm e_{k+1}(\bm\beta),
\ean
which implies that
\bna
\sigma_{t+1,i}(\bm\beta)=\sum^n_{j=1}\sum^{t}_{k=0}\alpha^{t-k}a^{(t+1-k)}_{ij}
(y_{k+1,j}-{\bm\beta}^T\bm\varphi_{k,j})^2, \label{least2}
\ena
where  $a^{(t+1-k)}_{ij}$ is the $i$-th row, $j$-th column entry of the matrix $\mathcal{A}^{t+1-k}$.

By minimizing the local cost function  $\sigma_{t+1,i}(\bm\beta)$ in   (\ref{least2}), we obtain the distributed FFLS estimate $\bm{\hat\theta}_{t+1,i}$ of the unknown time-varying parameter for sensor $i$, i.e.,
\bna
\bm{\hat\theta}_{t+1,i}&\triangleq& \arg\min_{\bm\beta}\sigma_{t+1,i}(\bm\beta)\nonumber\\
&=&\left[\sum^n_{j=1}\sum^t_{k=0}\alpha^{t-k}a^{(t+1-k)}_{ij}\bm\varphi_{k,j}\bm\varphi^T_{k,j}\right]^{-1}\nonumber\\
&&\left(\sum^n_{j=1}
\sum^t_{k=0}\alpha^{t-k}a^{(t+1-k)}_{ij}\bm\varphi_{k,j}y_{k+1,j}\right).\label{theta}
\ena
Denote
$
\bm P_{t+1,i}=\left(\sum^n_{j=1}\sum^t_{k=0}\alpha^{t-k}a^{(t+1-k)}_{ij}\bm\varphi_{k,j}\bm\varphi^T_{k,j}\right)^{-1}.
$
Then we write it into the following recursive form,
\bna
\bm P^{-1}_{t+1,i}=\sum_{j\in \mathcal{N}_i}a_{ij}(\alpha{{\bm P}}^{-1}_{t,j}+\bm\varphi_{t,j}\bm\varphi^T_{t,j}).\label{P_inverse}
\ena
By (\ref{theta}), we similarly have
\bna
\bm{\hat\theta}_{t+1,i}=\bm P_{t+1,i}\sum_{j\in \mathcal{N}_i}a_{ij}(\alpha\bm P^{-1}_{t,j}\bm{\hat\theta}_{t,j}+\bm \varphi_{t,j}y_{t+1,j}).
\label{theta1}
\ena
Note that in the above derivation, we assume that the matrix $\sum^n_{j=1}\sum^t_{k=0}\alpha^{t-k}a^{(t+1-k)}_{ij}\bm\varphi_{k,j}\bm\varphi^T_{k,j}$ is invertible which is usually not satisfied for small $t$. To solve this problem, we take the initial  matrix $\bm P_{0,i}$ to be positive definite. Then (\ref{P_inverse}) can be modified into the following equation,
\begin{align}
\bm P_{t+1,i}=&\Bigg(\sum^n_{j=1}\sum^t_{k=0}\alpha^{t-k}a^{(t+1-k)}_{ij}\bm\varphi_{k,j}\bm\varphi^T_{k,j}
\nonumber\\
&+\sum^n_{j=1}\alpha^{t+1}a^{(t+1)}_{ij}\bm P^{-1}_{0,j}\Bigg)^{-1}.\label{sparse16}
\end{align}
Though, the estimate given by (\ref{theta1}) has a slight difference with (\ref{theta}), which does not affect the analysis of the asymptotic properties of the estimates.

To design the distributed algorithm, we denote
\bna
\bm{\bar P}^{-1}_{t+1,i}=\alpha{{\bm P}}^{-1}_{t,i}+\bm\varphi_{t,i}\bm\varphi^T_{t,i}.\label{adap6}
\ena
By Lemma \ref{wl1}, we have
$\bm{\bar P}_{t+1,i}=\frac{1}{\alpha}(\bm P_{t,i}-\frac{\bm P_{t,i}\bm\varphi_{t,i}\bm\varphi^T_{t,i}\bm P_{t,i}}{\alpha+\bm\varphi^T_{t,i}\bm P_{t,i}\bm\varphi_{t,i}})$. Hence,
\ban
\bm{\bar\theta}_{t+1,i}&\triangleq&\bm{\bar P}_{t+1,i}(\alpha\bm P^{-1}_{t,i}\bm{\hat\theta}_{t,i}+\bm \varphi_{t,i}y_{t+1,i})\\
&=&\bm{\hat\theta}_{t,i}+\frac{\bm P_{t,i}\bm\varphi_{t,i}}{\alpha+\bm\varphi^T_{t,i}\bm P_{t,i}\bm\varphi_{t,i}}(y_{t+1,i}-\bm\varphi^T_{t,i}\bm{\hat\theta}_{t,i}).
\ean
Therefore, we get the following distributed FFLS algorithm of diffusion type, i.e., Algorithm \ref{algorithm2}.
\begin{algorithm}[!ht]
		\caption{Distributed FFLS algorithm  \label{algorithm2}}
		\hspace*{0.02in} {\bf Input:}
		$\{\bm\varphi_{t,i}, y_{t+1,i}\}^n_{i=1}$, $t=0,1,2,\cdots$\\
		\hspace*{0.02in} {\bf Output:}
		$\{\bm{\hat\theta}_{t+1,i}\}^n_{i=1}$, $t=0,1,2,\cdots$
\begin{algorithmic}[0]   
		\State {\bf  Initialization:}
	For each sensor $i\in\{1,\cdots,n\}$, begin with an initial vector $\bm {\hat\theta}_{0,i}$ and  an
	 initial positive definite matrix $\bm P_{0,i}>0$.

	\For{ \rm{each time} $t=0,1,2,\cdots$}
	\For{ \rm{each~ sensor~} $i=1,\cdots,n$}
	\State $\mathbf{Step\ 1.}$ Adaption (generate $\bm{\bar\theta}_{t+1,i}$ and $\bm{\bar P}_{t+1,i}$ based \hspace*{0.4in} on $\bm{\hat\theta}_{t,i}$, $\bm P_{t,i}$, $\bm\varphi_{t,i}$ and $y_{t+1,i}$):
	\begin{align}
		~~~~~~~~~~\bm{\bar\theta}_{t+1,i}&=\bm{\hat\theta}_{t,i}+\frac{\bm P_{t,i}\bm\varphi_{t,i}}{\alpha+\bm\varphi^T_{t,i}\bm P_{t,i}\bm\varphi_{t,i}}(y_{t+1,i}-\bm\varphi^T_{t,i}\bm{\hat\theta}_{t,i}),\label{adap1}\\
		\bm{\bar P}_{t+1,i}&=\frac{1}{\alpha}\left(\bm P_{t,i}-\frac{\bm P_{t,i}\bm\varphi_{t,i}\bm\varphi^T_{t,i}\bm P_{t,i}}{\alpha+\bm\varphi^T_{t,i}\bm P_{t,i}\bm\varphi_{t,i}}\right),\label{adap2}
	\end{align}
	
	\State $\mathbf{Step\ 2.}$  Combination  (generate $\bm P^{-1}_{t+1,i}$ and $\bm{\hat\theta}_{t+1,i}$  \hspace*{0.4in} by a convex combination of $\bm{\bar\theta}_{t+1,j}$ and $\bm{\bar P}_{t+1,j}$):
	\begin{align}
		\bm P^{-1}_{t+1,i}&=\sum_{j\in\mathcal{N}_i}a_{ij}\bm{\bar P}^{-1}_{t+1,j},\label{adap3}\\
		\bm{\hat\theta}_{t+1,i}&=\bm P_{t+1,i}\sum_{j\in \mathcal{N}_i}a_{ij}\bm{\bar P}^{-1}_{t+1,j}\bm{\bar\theta}_{t+1,j}.\label{adap4}
	\end{align}
	\EndFor
	\EndFor
\end{algorithmic}
\end{algorithm}

Note that when $\mathcal{A} = \bm I_n$, the  distributed FFLS algorithm will degenerate
to the classical FFLS (i.e., Algorithm \ref{algorithm1}), and when $\alpha=1$, the distributed FFLS algorithm will degenerate
to the distributed LS in \cite{Xie2021} which is used to estimate the time-invariant parameter. The quantity $1-\alpha$ is usually referred to as the speed of adaption. Intuitively, when the parameter process $\{\bm\theta_t\}$ is slowly time-varying, the adaptation speed should also be slow (i.e., $\alpha$ is large).
The purpose of this paper is to establish the stability  of the above diffusion FFLS-based adaptive filter  without independence or stationarity assumptions on random regression vector $\{\bm\varphi_{t,i}\}$.

In order to analyze the distributed FFLS algorithm, we need to derive the estimation error equation. Denote $\bm {\widetilde\theta}_{t,i}\triangleq\bm\theta_t-\bm{\hat\theta}_{t,i}$, then from (\ref{adap3}) and (\ref{adap4}), we have
\begin{align}
&\bm {\widetilde\theta}_{t+1,i}=\bm\theta_{t+1}-\bm P_{t+1,i}\sum_{j\in \mathcal{N}_i}a_{ij}\bm{\bar P}^{-1}_{t+1,j}\bm{\bar\theta}_{t+1,j}\nonumber\\
=&\bm P_{t+1,i}\sum_{j\in \mathcal{N}_i}a_{ij}\bm{\bar P}^{-1}_{t+1,j}\bm{\theta}_{t+1}-\bm P_{t+1,i}\sum_{j\in \mathcal{N}_i}a_{ij}\bm{\bar P}^{-1}_{t+1,j}\bm{\bar\theta}_{t+1,j}\nonumber\\
=&\bm P_{t+1,i}\sum_{j\in \mathcal{N}_i}a_{ij}\bm{\bar P}^{-1}_{t+1,j}(\bm{\theta}_{t+1}-\bm{\bar\theta}_{t+1,j}).\label{adap5}
\end{align}

By (\ref{model}), (\ref{de1}), (\ref{adap1}) and (\ref{adap2}), we can obtain the following equation,
\begin{align}
&~~~~\bm{\theta}_{t+1}-\bm{\bar\theta}_{t+1,i}\nonumber\\
&=\bm{\theta}_{t}+\Delta\bm\theta_{t}
-\bm{\hat\theta}_{t,i}-\frac{\bm P_{t,i}\bm\varphi_{t,i}}{\alpha+\bm\varphi^T_{t,i}\bm P_{t,i}\bm\varphi_{t,i}}(y_{t+1,i}-\bm\varphi^T_{t,i}\bm{\hat\theta}_{t,i})\nonumber\\
&=\Big(\bm I_m-\frac{\bm P_{t,i}\bm\varphi_{t,i}\bm\varphi^T_{t,i}}{\alpha+\bm\varphi^T_{t,i}\bm P_{t,i}\bm\varphi_{t,i}}\Big)\bm{\widetilde\theta}_{t,i}-\frac{\bm P_{t,i}\bm\varphi_{t,i}w_{t+1,i}}{\alpha+\bm\varphi^T_{t,i}\bm P_{t,i}\bm\varphi_{t,i}}+\Delta\bm\theta_{t}\nonumber\\
&=\alpha\bm{\bar P}_{t+1,i}\bm P^{-1}_{t,i}\bm{\widetilde\theta}_{t,i}-\frac{\bm P_{t,i}\bm\varphi_{t,i}w_{t+1,i}}{\alpha+\bm\varphi^T_{t,i}\bm P_{t,i}\bm\varphi_{t,i}}+\Delta\bm\theta_{t}.\label{adap7}
\end{align}

For convenience of analysis, we introduce the following set of notations,
\ban &&\bm Y_t={\rm {col}}\{y_{t,1},\cdots,y_{t,n}\}, \hskip 1.94cm (n\times1)\nonumber\\
&& \bm \Phi_t= {\rm{diag}}\{\bm\varphi_{t,1},\cdots,\bm\varphi_{t,n}\},\hskip 1.5cm(mn\times n)\nonumber\\
&& \bm W_t={\rm {col}}\{w_{t,1},\cdots,w_{t,n}\},\hskip 1.65cm(n\times1)\nonumber\\
&&\bm P_t= {\rm{diag}}\{\bm P_{t,1},\cdots,\bm P_{t,n}\},\hskip 1.6cm(mn\times mn)\nonumber\\
&&\bm{\bar P}_t={\rm{diag}}\{\bm{\bar P}_{t,1},\cdots,\bm{\bar P}_{t,n}\},\hskip 1.6cm(mn\times mn)\nonumber\\
&&\bm{\Theta}_t={\rm {col}}\{\underbrace{\bm {\theta}_{t},\cdots,\bm {\theta}_{t}}_{n}\},\hskip 2.4cm (mn\times1)\nonumber\\
&&\Delta\bm\Theta_t={\rm {col}}\{\underbrace{\Delta\bm {\theta}_{t},\cdots,\Delta\bm\theta_{t}}_{n}\},\hskip 1.55cm (mn\times1)\nonumber\\
&&\bm{L}_t={\rm{diag}}\{\bm {L}_{t,1},\cdots,\bm {L}_{t,n}\},\hskip 1.66cm (mn\times n)\nonumber\\
&&\hskip 1cm {\rm where}~~ \bm {L}_{t,i}=\frac{\bm P_{t,i}\bm\varphi_{t,i}}{\alpha+\bm\varphi^T_{t,i}\bm P_{t,i}\bm\varphi_{t,i}},\nonumber\\
&&\bm{\widetilde\Theta}_t={\rm {col}}\{\bm {\widetilde\theta}_{t,1},\cdots,\bm {\widetilde\theta}_{t,n}\},\hskip 2cm (mn\times1)\nonumber\\
&&\mathscr{A}=\mathcal{A}\otimes \bm I_m,\hskip 3.68cm (mn\times mn)
\ean

Hence by (\ref{adap5}) and (\ref{adap7}), we have the following equation about estimation error,
\begin{align}
\bm{\widetilde\Theta}_{t+1}=\alpha\bm P_{t+1}\mathscr{A}\bm{ P}^{-1}_{t}\bm{\widetilde\Theta}_{t}-
\bm P_{t+1}\mathscr{A}\bm{\bar P}^{-1}_{t+1}(\bm L_t\bm W_{t+1}+\Delta\bm\Theta_t).\label{adap8}
\end{align}
 From (\ref{adap8}), we  see that the properties of product of random
 	matrices, i.e., $\prod_{t}\alpha\bm P_{t+1}\mathscr{A}\bm{ P}^{-1}_{t}$, play  important roles  in  stability analysis of the homogeneous part in error equation. 

As we all know,  the analysis of product of random
matrices is generally a difficult mathematical problem if the random matrices  do not satisfy the
independency or stationarity assumptions. 
There is some work to study this problem, which focuses on either symmetric random matrix or scalar gain case. For example, \cite{Gan2019} and \cite{Xie20181} investigated
	the  convergence of consensus-diffusion SG algorithm and the  stability of  consensus normalized LMS algorithm where  the random matrices in  error equations are symmetric.  Note that the  random matrices $\alpha\bm P_{t+1}\mathscr{A}\bm{ P}^{-1}_{t}$  here are asymmetric.
	Although \cite{Xie20182} studied the properties  of the asymmetric random matrices in the  LMS-based estimation error equation, the adaptive gain of distributed LMS algorithm in \cite{Xie20182}  is a scalar while the gain $\frac{\bm P_{t,i}}{\alpha+\bm\varphi^T_{t,i}\bm P_{t,i}\bm\varphi_{t,i}}$ in (\ref{adap1}) of this paper  is a random matrix.
	Hence the methods used in existing literature  including \cite{Xie20181,Gan2019,Xie20182} are no longer applicable to our case.
	One of the main purposes of this paper is to overcome
	the above difficulties by using both the specific structure of the diffusion FFLS and some 
	  results of FFLS on single sensor case (see \cite{Guo1994}).


\section{Stability of distributed FFLS algorithm under fixed undirected graph}\label{stability_undirected}
In this section, we will establish exponential stability  for the homogeneous part of the error equation (\ref{adap8}) and the tracking error bounds for the proposed distributed FFLS algorithm in Algorithm \ref{algorithm2} without requiring statistical independence on the system
signals. For this purpose, we need to introduce some definitions on the stability of random
matrices (see \cite{Guo1994}) and assumptions on the graph and random regression vectors.
\subsection{Some definitions}
\begin{definition}
A random matrix sequence
	$\{\bm A_t, t\geq0\}$ defined on the basic probability space $(\Omega, \mathscr{F}, P)$ is called $L_p$-stable $(p>0)$ if $\sup_{t\geq0}\mathbb{E}(\|\bm A_t\|^p)<\infty$, where $\mathbb{E}(\cdot)$ denotes the mathematical expectation operator.
	We define $\|\bm A_t\|_{L_p}\triangleq [\mathbb{E}(\|\bm A_t\|^p)]^{\frac{1}{p}}$ as the $L_p$-norm of the random matrix $\bm A_t$.
\end{definition}

\begin{definition}
	A sequence of $n\times n$ random matrices $\bm A=\{\bm A_t, t\geq0\}$ is called $L_p$-exponentially stable $(p\geq 0)$ with parameter $\lambda\in[0,1)$, if it belongs to the following set
	\begin{align}
		S_p(\lambda)=\Big\{&\bm A:\Big \|\prod^t_{j=k+1}\bm A_j\Big\|_{L_p}\leq M\lambda^{t-k}, \forall t\geq k, \nonumber\\
		 &\forall k\geq0, {\rm for ~ some }~M>0\Big\}.\label{adap9}
	\end{align}
\end{definition}

As demonstrated by Guo in \cite{Guo1994}, $\{\bm A_t,t\geq0\}\in S_p(\lambda)$ is in some sense the necessary and
sufficient condition for stability of $\{\bm x_t\}$
generated by $\bm x_t = \bm A_t \bm x_t+\bm \xi_{t+1},~t\geq0$. Also,
the stability analysis of the matrix sequence may be reduced to
that of a certain class of scalar sequence, which can be further analyzed based on some excitation conditions on the regressors. To this end, we introduce the following subset of $S_1(\lambda)$  for a scalar
sequence $a=(a_t, t\geq0)$.
\begin{align*}
	S^0(\lambda)=\Big\{&a: a_t\in[0,1), \mathbb{E}\left(\prod^t_{j=k+1}a_j\right)\leq M\lambda^{t-k}, \\
	&\forall t\geq k,
	 \forall k\geq0, {\rm for ~ some}~M>0\Big\}.
\end{align*}
The definition $S^0(\lambda)$ will be used  when we convert the product of a random matrix to that of a scalar sequence.

\begin{remark}\label{remark1}
It is clear that if there exist a constant
$a_0\in(0, 1)$ such that $a_t \leq a_0$ for all $t$, then $a_t \in S^0(a_0)$.  More
properties about the set $S^0(\lambda)$ can be found in \cite{Guobook2020}.
\end{remark}

\subsection{Assumptions}
\begin{assumption}\label{a1}
	The undirected graph $\mathcal{G}$ is connected.
\end{assumption}
\begin{remark}\label{remark31} For any $k>1$, we denote $\mathcal{A}^k\triangleq(a_{ij}^{(k)})$ with $\mathcal{A}$ being the weighted adjacency matrix of the graph $\mathcal{G}$,  i.e., $a_{ij}^{(k)}$
	is the $i$-th row, $j$-th column element of
	the matrix $\mathcal{A}^k$.  Under Assumption \ref{a1}, it is clear that $\mathcal{A}^k$ is a positive matrix for $k\geq D_{\mathcal{G}}$, which means that  $a_{ij}^{(k)}>0$ for any $i$ and $j$ (cf.,  \cite{Godsil2001}).
\end{remark}

\begin{assumption}[Cooperative Excitation Condition]\label{a2}
	For the adapted sequences $\{\bm\varphi_{t,i}, \mathscr{F}_t, t\geq0\}$, where $\mathscr{F}_t$
	is a sequence of non-decreasing $\sigma$-algebras,  there exists an integer $h>0$ such that
	$\{1-\lambda_t\}\in S^0(\lambda)$ for some $\lambda\in(0,1)$, where $\lambda_t$ is defined by
	\ban
	\lambda_t\triangleq\lambda_{\min}\left[\mathbb{E}\left(\frac{1}{n(1+h)}\sum^n_{i=1}
	\sum^{(t+1)h}_{k=th+1}\frac{\bm\varphi_{k,i}\bm\varphi^T_{k,i}}{1+\|\bm\varphi_{k,i}\|^2}
	\Big|\mathscr{F}_{th}\right)\right]
	\ean
	with $\mathbb{E}(\cdot| \cdot)$ being  the conditional mathematical
	expectation operator.
\end{assumption}
\begin{remark}
	Assumption \ref{a2} is also  used to guarantee the stability and performance of the distributed LMS algorithm (see e.g., \cite{Xie20181,Xie20182}). 
	We give some intuitive explanations for the above cooperative excitation condition about the following two aspects.
	
	(1) ``Why excitation". Let us consider an extreme case where all regression
	vectors $\bm \varphi_{k,i}$ are equal to zero, then Assumption \ref{a2} can  not be satisfied. Moreover, from (\ref{model}), we see that the unknown parameter $\bm{\theta}_t$ can not be estimated or tracked since the observations $y_{t,i}$ 
	do not contain any information about the unknown parameter $\bm\theta_t$. In order
	to estimate $\bm{\theta}_t$,  some nonzero information condition (named excitation condition) should be imposed
	on the regression vectors $\bm \varphi_{t,i}$. In fact, Assumption \ref{a2} intuitively gives a lower bound (which may be changed over time) of the sequence $\{\lambda_t\}$. For example,
	if there exists a constant $\lambda_0\in(0,1)$ such that $\inf_t\lambda_t\geq\lambda_0$, then by Remark \ref{remark1}, we know that   Assumption \ref{a2} can be satisfied.

(2) ``Why cooperative".   Compared with the excitation condition for FFLS algorithm of single sensor case in \cite{Guo1994}, i.e., there exists a constant $h>0$ such that
\begin{gather}
\{1-{\lambda}{'}_t, t\geq 0\}\in S^0({\lambda}{'})\label{single_excitation}
\end{gather}
 for some ${\lambda}{'}$ where
\begin{gather*}
{\lambda}{'}_t=\lambda_{\min}\left[\mathbb{E}\left(\frac{1}{1+h}
\sum^{(t+1)h}_{k=th+1}\frac{\bm\varphi_{k,i}\bm\varphi^T_{k,i}}{1+\|\bm\varphi_{k,i}\|^2}
\Big|\mathscr{F}_{th}\right)\right].
\end{gather*}
	Assumption \ref{a2} contains not only temporal union information but also spatial union information of all the sensors, which means that Assumption \ref{a2} is much weaker than the condition (\ref{single_excitation}) since $\lambda_t\geq{\lambda}{'}_t$ when $n>1$. Besides, we also note that Assumption \ref{a2} can be reduced  to the condition  (\ref{single_excitation}) when $n=1$. In fact,   Assumption \ref{a2}  can reflect the cooperative effect of multiple sensors in the sense that the estimation task can be still fulfilled by the cooperation of multiple sensors even if any of them cannot.
\end{remark}

\subsection{Main results}\label{results}

In order to establish  exponential stability of the product of random matrices $\alpha\bm P_{t+1}\mathscr{A}\bm{ P}^{-1}_{t}$, we first analyze the properties of the random matrix $\bm P_t$ to obtain its upper bound.
\begin{lemma}\label{lemma1}
	For $\{\bm P_t\}$ generated by (\ref{adap2}) and (\ref{adap3}), under Assumptions \ref{a1}-\ref{a2}, we have
	\begin{align}
	T_{t+1}
	\leq\frac{1}{\alpha^{h'}}(1-\beta_{t+1})(h'-D_{\mathcal{G}})tr(\bm P_{th'+1}).
	\end{align}
	where
	\begin{align*}
	T_t&\triangleq\sum^{th'}_{k=(t-1)h'+D_{\mathcal{G}}+1}tr(\bm P_{k+1}), ~T_0=0,\\
	\beta_{t+1}&\triangleq\frac{a^2_{\min}\gamma_{t+1}}{n(h'-D_{\mathcal{G}})\left(\alpha^{h'}
		+\lambda_{\max}\left(\sum^n_{l=1}\bm P_{th'+1,l}\right)\right)tr(\bm P_{th'+1})},\\
	\gamma_{t+1}&\triangleq tr\left(\left(\sum^n_{l=1}\bm P_{th'+1,l}\right)^2\sum^{(t+1)h'}_{k=th'+D_{\mathcal{G}}+1}\sum^n_{j=1}
	\frac{\bm\varphi_{k,j}\bm\varphi^T_{k,j}}{(1+\|\bm\varphi_{k,j}\|^2)}\right),\\
	a_{\min}&\triangleq\min\limits_{i,j\in\{1,\cdots,n\}}a^{(D_{\mathcal{G}})}_{ij}>0,\\
	h'&\triangleq 2h+D_{\mathcal{G}},
	\end{align*}
	and  $h$ is given by Assumption \ref{a2}.
\end{lemma}
\begin{proof}
Note that $a^{(k)}_{ij}$ is the $i$-th
row, $j$-th column element of the matrix $\mathcal{A}^k,$ $k\geq1$, where $a^{(1)}_{ij}=a_{ij}$.
By (\ref{adap6}), we have $\bm P^{-1}_{k+1,i}\geq\sum^n_{j=1}a_{ij}{\alpha\bm P^{-1}_{k,j}}$. Hence by the inequality
\bna
\Big(\sum^n_{j=1}a_{ij}{\bm A_j}\Big)^{-1}\leq\sum^n_{j=1}a_{ij}{\bm A^{-1}_j}\label{adap10}
\ena with $\bm A_j\geq0$,
we obtain for any $t\geq0$, and any $k\in[th'+D_{\mathcal{G}}+1,(t+1)h']$,
\begin{align}
\bm P_{k,i}&\leq\Big(\sum^n_{j=1}a_{ij}{\alpha\bm P^{-1}_{k-1,j}}\Big)^{-1}
\leq\frac{1}{\alpha}\sum^n_{j=1}a_{ij}{\bm P_{k-1,j}}\nonumber\\
&\leq\frac{1}{\alpha}\sum^n_{j=1}a_{ij}\left(\frac{1}{\alpha}\sum^n_{l=1}a_{jl}\bm P_{k-2,l}\right)\nonumber\\
&=\frac{1}{\alpha^2}\sum^n_{j=1}a^{(2)}_{ij}\bm P_{k-2,j}
\leq\cdots\nonumber\\
&\leq\frac{1}{\alpha^{k-th'-1}}\sum^n_{j=1}a^{({k-th'-1})}_{ij}\bm P_{th'+1,j}\nonumber\\
&\leq\frac{1}{\alpha^{h'-1}}\sum^n_{j=1}a^{({k-th'-1})}_{ij}\bm P_{th'+1,j}.\label{adap16}
\end{align}
Denote $\bm Q^{k, th'}_{i}=\sum^n_{j=1}a^{({k-th'-1})}_{ij}\bm P_{th'+1,j}$. Then by (\ref{adap6}), (\ref{adap3}), (\ref{adap10}) and (\ref{adap16}), we have for $k\in[th'+D_{\mathcal{G}}+1,(t+1)h']$,
\begin{align}
\bm{P}_{k+1,i}&=\left(\sum^n_{j=1}a_{ij}(\alpha{{\bm P}}^{-1}_{k,j}+\bm\varphi_{k,j}\bm\varphi^T_{k,j})\right)^{-1}\nonumber\\
&\leq\sum^n_{j=1}a_{ij}(\alpha{{\bm P}}^{-1}_{k,j}+\bm\varphi_{k,j}\bm\varphi^T_{k,j})^{-1}\nonumber\\
&\leq\sum^n_{j=1}a_{ij}\left(\alpha\left(\frac{1}{\alpha^{h'-1}}\bm Q^{k, th'}_{j}\right)^{-1}+\bm\varphi_{k,j}\bm\varphi^T_{k,j}\right)^{-1}.\label{adap11}
\end{align}
By Lemma \ref{wl1} and (\ref{adap11}), it follows that
\begin{align}
\bm{P}_{k+1,i}&\leq\frac{1}{\alpha^{h'}}\sum^n_{j=1}a_{ij}\Bigg(\bm Q^{k, th'}_{j}-\frac{\bm Q^{k, th'}_{j}\bm\varphi_{k,j}\bm\varphi^T_{k,j}\bm Q^{k, th'}_{j}}{\alpha^{h'}+\bm\varphi^T_{k,j}\bm Q^{k, th'}_{j}\bm\varphi_{k,j}}\Bigg)\nonumber\\
&=\frac{1}{\alpha^{h'}}\sum^n_{j=1}a^{({k-th'})}_{ij}\bm P_{th'+1,j}\nonumber\\
&~~~~-\frac{1}{\alpha^{h'}}\sum^n_{j=1}a_{ij}\frac{\bm Q^{k, th'}_{j}\bm\varphi_{k,j}\bm\varphi^T_{k,j}\bm Q^{k, th'}_{j}}{\alpha^{h'}+\bm\varphi^T_{k,j}\bm Q^{k, th'}_{j}\bm\varphi_{k,j}}\nonumber\\
&\leq\frac{1}{\alpha^{h'}}\sum^n_{j=1}a^{({k-th'})}_{ij}\bm P_{th'+1,j}\nonumber\\
&~~~~-\frac{1}{\alpha^{h'}}\sum^n_{j=1}\frac{a_{ij}\bm Q^{k, th'}_{j}\bm\varphi_{k,j}\bm\varphi^T_{k,j}\bm Q^{k, th'}_{j}}{\alpha^{h'}+\lambda_{\max}(\bm Q^{k, th'}_{j})(1+\|\bm\varphi_{k,j}\|^2)}.\label{adap12}
\end{align}
Then by (\ref{adap12}), we have
\begin{align*}
&tr(\bm P_{k+1})=tr\Bigg(\sum^n_{i=1}\bm P_{k+1,i}\Bigg)\\
\leq&\frac{1}{\alpha^{h'}}tr\Bigg(\sum^n_{i=1}\sum^n_{j=1}a^{({k-th'})}_{ij}\bm P_{th'+1,j}\Bigg)\nonumber\\
&-\frac{1}{\alpha^{h'}}tr\Bigg(\sum^n_{i=1}\sum^n_{j=1}a_{ij}\frac{\bm Q^{k, th'}_{j}\bm\varphi_{k,j}\bm\varphi^T_{k,j}\bm Q^{k, th'}_{j}}{\alpha^{h'}+\lambda_{\max}(\bm Q^{k, th'}_{j})(1+\|\bm\varphi_{k,j}\|^2)}\Bigg)\\
=&\frac{1}{\alpha^{h'}}\Bigg(tr(\bm P_{th'+1})-\sum^n_{j=1}\frac{tr\left(\bm Q^{k, th'}_{j}\bm\varphi_{k,j}\bm\varphi^T_{k,j}\bm Q^{k, th'}_{j}\right)}{\alpha^{h'}+\lambda_{\max}(\bm Q^{k, th'}_{j})(1+\|\bm\varphi_{k,j}\|^2)}\Bigg).
\end{align*}
Hence combining this with  the inequality $\sum^{n}_{j=1}\frac{a_j}{b_j}\geq\frac{\sum^{n}_{j=1}a_j}{\sum^{n}_{j=1}b_j}$ where
$a_j\geq0$ and $b_j\geq0$, we obtain that
\begin{align}
&tr(\bm P_{k+1})\nonumber\\
\leq&\frac{1}{\alpha^{h'}}\left(tr(\bm P_{th'+1})-\frac{tr\left(\sum^n_{j=1}\left(\bm Q^{k, th'}_{j}\right)^2\frac{\bm\varphi_{k,j}\bm\varphi^T_{k,j}}{(1+\|\bm\varphi_{k,j}\|^2)}\right)}{\sum^n_{j=1}\left(\alpha^{h'}+\lambda_{\max}\left(\bm Q^{k, th'}_{j}\right)\right)}\right).\label{adap13}
\end{align}
By Remark \ref{remark31}, we know that  $a^{(k)}_{ij}\geq a_{\min}$ holds for all
$k\geq D_{\mathcal{G}}$.
Thus, by (\ref{adap13}), we have for $k\in[th'+D_{\mathcal{G}}+1,(t+1)h']$
\begin{align}
tr(\bm P_{k+1})
&\leq\frac{1}{\alpha^{h'}}\Bigg(tr(\bm P_{th'+1})\nonumber\\
&-\frac{a^2_{\min}tr\left(\sum^n_{j=1}\left(\sum^n_{l=1}\bm P_{th'+1,l}\right)^2\frac{\bm\varphi_{k,j}\bm\varphi^T_{k,j}}{(1+\|\bm\varphi_{k,j}\|^2)}\right)}{n\left(\alpha^{h'}
	+\lambda_{\max}\left(\sum^n_{l=1}\bm P_{th'+1,l}\right)\right)}\Bigg).\label{adap14}
\end{align}
Summing up both sides of (\ref{adap14}) from $th'+D_{\mathcal{G}}+1$ to $(t+1)h'$,
by the definition of $\beta_{t+1}$,
we have
\begin{align*}
T_{t+1}&=\sum^{(t+1)h'}_{k=th'+D_{\mathcal{G}}+1}tr(\bm P_{k+1})\\
&\leq\frac{1}{\alpha^{h'}}(1-\beta_{t+1})(h'-D_{\mathcal{G}})tr(\bm P_{th'+1}).
\end{align*}
This completes the proof of the lemma.
\end{proof}

Before giving the boundness of the random matrix $\bm P_t$, we first introduce two lemmas in \cite{Guo1994}.

\begin{lemma}\label{l1}\cite{Guo1994}
	Let $\{1-\xi_t\}\in S^0(\lambda)$, and $0<\xi_t\leq \xi^*<1$, where $\xi^*$ is a positive constant. Then for any $\varepsilon\in(0,1)$,
	$\{1-\varepsilon\xi_t\}\in S^0(\lambda^{(1-\xi^*)\varepsilon})$.
\end{lemma}

\begin{lemma}\label{adap25}\cite{Guo1994}
	Let $\{x_t, \mathscr{F}_t\}$ be an adapted process, and
	\ban
	x_{t+1}\leq \xi_{t+1}x_t+\eta_{t+1},~~~~t\geq0, ~~\mathbb{E}x^2_0<\infty,
	\ean
	where $\{\xi_t,\mathscr{F}_t\}$ and $\{\eta_t,\mathscr{F}_t\}$ are two adapted nonnegative process with properties:
	\begin{gather*}
		\xi_t\geq\varepsilon_0>0, ~~\forall t,\\
		\mathbb{E}(\eta^2_{t+1}|\mathscr{F}_t)\leq N<\infty, ~~\forall t,\\
		\left\|\prod^t_{k=j}\mathbb{E}(\xi^4_{k+1}|\mathscr{F}_k)\right\|\leq M\eta^{t-j+1},~~ \forall t\geq j,~~\forall j,
	\end{gather*}
	where $\varepsilon_0,M,N$ and $\eta\in(0,1)$ are constants. Then we have
	\ban
	&(i)&~~\left\|\prod^t_{k=j}\xi_k\right\|_{L_2}\leq M^{\frac{1}{4}}\eta^{\frac{1}{4}(t-j
		+1)}, ~~~~\forall t\geq j, ~~\forall j;\\
	&(ii)&~~\sup_{t}\mathbb{E}(\|x_t\|)<\infty.
	\ean
\end{lemma}

The following lemma proves the  boundedness of the random matrix sequence $\{\bm P_t\}$.
\begin{lemma}\label{yinli2}
	For $\{\bm P_t\}$ generated by (\ref{adap2}) and (\ref{adap3}), under Assumptions \ref{a1}-\ref{a2}, we have for any $p\geq 1$, $\bm P_t$ is $L_p$ stable, i.e.,
	\ban
	\sup_{t\geq0}\mathbb{E}(\|\bm P_t\|^p)<\infty
	\ean
	provided that $\lambda^{\frac{a^2_{\min}}{32pmh(4h+D_{\mathcal{G}}-1)}}<\alpha<1$, where $\lambda$ and $h$ are given by Assumption \ref{a2}, and $m$ is the dimension of $\bm\varphi_{t,i}$.
\end{lemma}
\begin{proof}
	For any $t\geq0$, there exists an integer $z_t=\lfloor\frac{th'+D_{\mathcal{G}}}{h}\rfloor+1$ such that
	\bna
	(z_t-1)h\leq(th'+D_{\mathcal{G}}+1)\leq z_t h+1.\label{adap23}
	\ena
	By the definition of $\beta_{t+1}$ in Lemma \ref{lemma1}, it is clear that
	\begin{align}
	&\beta_{t+1}\nonumber\\
	\geq&\frac{a^2_{\min}tr\Big(\left(\sum^n_{l=1}\bm P_{th'+1,l}\right)^2\sum^{(z_t+1)h}_{k=z_th+1}\sum^n_{j=1}
		\frac{\bm\varphi_{k,j}\bm\varphi^T_{k,j}}{(1+\|\bm\varphi_{k,j}\|^2)}\Big)}{n(h'-D_{\mathcal{G}})\left(\alpha^{h'}
		+\lambda_{\max}\left(\sum^n_{l=1}\bm P_{th'+1,l}\right)\right)tr(\bm P_{th'+1})}\nonumber\\
	\triangleq &b_{t+1}. \label{adap15}
	\end{align}
Hence by Lemma \ref{lemma1} and (\ref{adap15}), we obtain
\bna
T_{t+1}
\leq\frac{1}{\alpha^{h'}}(1-b_{t+1})(h'-D_{\mathcal{G}})tr(\bm P_{th'+1}).\label{adap17}
\ena
By the inequality $\bm P_{k,i}\leq\frac{1}{\alpha}\sum^n_{j=1}a_{ij}{\bm P_{k-1,j}}$ used in (\ref{adap16}) it follows that
\ban
&&(h'-D_{\mathcal{G}})tr(\bm P_{th'+1})
=\sum^{th'}_{k=(t-1)h'+D_{\mathcal{G}}+1}tr(\bm P_{th'+1})\\
&=&\sum^{th'}_{k=(t-1)h'+D_{\mathcal{G}}+1}\sum^n_{i=1}tr(\bm P_{th'+1,i})\\
&\leq&\sum^{th'}_{k=(t-1)h'+D_{\mathcal{G}}+1}\sum^n_{i=1}tr\left(\frac{1}{\alpha^{th'-k}}\sum^{n}_{j=1}a^{(th'-k)}_{ij}\bm P_{k+1,j}\right)\\
&\leq&\frac{1}{\alpha^{h'-D_{\mathcal{G}}-1}}\sum^{th'}_{k=(t-1)h'+D_{\mathcal{G}}+1}tr(\bm P_{k+1})
=\frac{1}{\alpha^{h'-D_{\mathcal{G}}-1}}T_t.
\ean
Hence by (\ref{adap17}), we have
\bna
T_{t+1}
\leq\frac{1}{\alpha^{2h'-D_{\mathcal{G}}-1}}(1-b_{t+1})T_t.\label{adap18}
\ena
For  $p\geq1$, denote
\bna
c_{t+1}=\frac{1}{\alpha^{p(2h'-D_{\mathcal{G}}-1})}\left(1-\frac{b_{t+1}}{2}\right)I_{\{tr(\bm P_{th'+1})\geq1\}}\label{adap19}
\ena
where $I_{\{\cdot\}}$ denotes  the indicator
function, whose value is 1 if its argument (a formula) is true, and 0, otherwise.
Then by (\ref{adap17}) and (\ref{adap18}), we have
\bna
T^p_{t+1}&\leq& T^p_{t+1}\left(I_{\{tr(\bm P_{th'+1})\geq1\}}+I_{\{tr(\bm P_{th'+1})\leq1\}}\right)\nonumber\\
&\leq&\frac{1}{\alpha^{p(2h'-D_{\mathcal{G}}-1)}}(1-b_{z_t+1})^p T^p_tI_{\{tr(\bm P_{th'+1})\geq1\}}\nonumber\\
&&+T^p_{t+1}I_{\{tr(\bm P_{th'+1})\leq1\}}\nonumber\\
&\leq&c_{t+1}T^p_t+\frac{1}{\alpha^{ph'}}(h'-D_{\mathcal{G}})^p.\label{adap20}
\ena
Denote
$$\bm H_{z_t}=\mathbb{E}\left(\sum^{(z_t+1)h}_{k=z_th+1}\sum^n_{j=1}
\frac{\bm\varphi_{k,j}\bm\varphi^T_{k,j}}{1+\|\bm\varphi_{k,j}\|^2}\Bigg|\mathscr{F}_{z_th}\right).$$
By the inequality $$tr\left(\left(\sum^n_{l=1}\bm P_{th'+1,l}\right)^2\right)\geq m^{-1}\left(tr \left(\sum^n_{l=1}\bm P_{th'+1,l}\right)\right)^2$$ and $\bm P_{th'+1,l}\in\mathscr{F}_{th'}\subset\mathscr{F}_{z_th}$, from the definition of $b_{t+1}$ in (\ref{adap15}), we can conclude the following inequality,
\begin{align}
	&\mathbb{E}(b_{t+1}|\mathscr{F}_{z_th})\nonumber\\
	=&\frac{a^2_{\min}tr\left[\left(\sum^n_{l=1}\bm P_{th'+1,l}\right)^2\bm H_{z_t}
		\right]}{n(h'-D_{\mathcal{G}})\left(\alpha^{h'}
		+\lambda_{\max}\left(\sum^n_{l=1}\bm P_{th'+1,l}\right)\right)tr(\bm P_{th'+1})}\nonumber\\
	\geq& \frac{a^2_{\min}\left(tr(\bm P_{th'+1})\right)^2\lambda_{\min}(\bm H_{z_t})}{mn(h'-D_{\mathcal{G}})\left(\alpha^{h'}
		+\lambda_{\max}\left(\sum^n_{l=1}\bm P_{th'+1,l}\right)\right)tr(\bm P_{th'+1})}\nonumber\\
	\geq& \frac{a^2_{\min}\left(tr(\bm P_{th'+1})\right)\lambda_{z_t}(1+h)}{m(h'-D_{\mathcal{G}})\left(\alpha^{h'}
		+\lambda_{\max}\left(\sum^n_{l=1}\bm P_{th'+1,l}\right)\right)}\nonumber\\
	\geq& \frac{a^2_{\min}\left(tr(\bm P_{th'+1})\right)\lambda_{z_t}(1+h)}{m(h'-D_{\mathcal{G}})\left(1
		+tr(\bm P_{th'+1})\right)}\nonumber\\
	\geq&\frac{a^2_{\min}\lambda_{z_t}(1+h)}{2m(h'-D_{\mathcal{G}})}~~~~ {\rm{on }}~~\{tr(\bm P_{th'+1})\geq1\}. \label{keystep}
\end{align}
Hence by the definition of $c_{z_t+1}$ in (\ref{adap19}),
\begin{align}
&\mathbb{E}(c_{t+1}|\mathscr{F}_{z_th})\nonumber\\
\leq&\frac{1}{\alpha^{p(2h'-D_{\mathcal{G}}-1)}}
\left(1-\frac{a^2_{\min}\lambda_{z_t}(1+h)}{4m(h'-D_{\mathcal{G}})}\right)I_{\{tr(\bm P_{th'+1})\geq1\}}.\label{adap22}
\end{align}
Denote
\ban
d_{t+1}=
\begin{cases}
	c_{t+1}, &tr(\bm P_{th'+1})\geq1;\\
	\frac{1}{\alpha^{p(2h'-D_{\mathcal{G}}-1)}}
	\left(1-\frac{a^2_{\min}\lambda_{z_t}(1+h)}{4m(h'-D_{\mathcal{G}})}\right),& {\rm otherwise.}
\end{cases}
\ean
Then by (\ref{adap20}) and (\ref{adap22}), we have
\bna
T^p_{t+1}\leq&d_{t+1}T^p_t+\frac{1}{\alpha^{ph'}}(h'-D_{\mathcal{G}})^p.\label{adap21}
\ena
Since  $\lambda_{z_t}\leq\frac{h}{1+h}$ and $b_{t+1}\leq\frac{a^2_{\min}h}{h'-D_{\mathcal{G}}}$, we know that
$d_{t+1}\geq\varepsilon_0$ with $\varepsilon_0$ being a positive constant.
Denote $\mathscr{B}_t\triangleq\mathscr{F}_{z_th}$, then by the definition of $z_t$, it is clear that
$z_{t+1}\geq z_t+2$. Thus,
we obtain that $d_{t+1}\in\mathscr{F}_{(z_t+1)h}\subset\mathscr{B}_{t+1}$.
Similar to the analysis of $(\ref{adap22})$, we have
\bna
\mathbb{E}(c^4_{t+1}|\mathscr{B}_t)\leq\frac{1}{\alpha^{4p(2h'-D_{\mathcal{G}}-1)}}
\left(1-\frac{a^2_{\min}\lambda_{z_t}(1+h)}{4m(h'-D_{\mathcal{G}})}\right).
\ena
Hence by the definition of $d_{t+1}$, it follows that
\begin{align}
&\Big\|\prod^t_{k=j}\mathbb{E}(d^4_{k+1}|\mathscr{B}_k)\Big\|_{L_1}\nonumber\\
\leq&\Big\|\prod^t_{k=j}\left(\frac{1}{\alpha^{4p(2h'-D_{\mathcal{G}}-1)}}
\left(1-\frac{a^2_{\min}\lambda_{z_k}(1+h)}{8mh}\right)\right)\Big\|_{L_1}.\label{adap24}
\end{align}
By Assumption \ref{a2} and the fact $\lambda_{z_k}\leq\frac{h}{1+h}$, applying Lemma \ref{l1}, we obtain
$\{1-\frac{a^2_{\min}\lambda_{z_k}(1+h)}{8mh}\}\in S^0\Big(\lambda^{\frac{a^2_{\min}}{8mh}}\Big)$. By
(\ref{adap24}), we see that there exists a positive constant $N $
such that
\ban
\Big\|\prod^t_{k=j}\mathbb{E}(d^4_{k+1}|\mathscr{B}_k)\Big\|_{L_1}\leq
N\lambda_1^{t-j+1},
\ean
where $\lambda_1=\frac{1}{\alpha^{4p(2h'-D_{\mathcal{G}}-1)}}\lambda^{\frac{a^2_{\min}}{8mh}}\in(0,1)$. Furthermore, by Lemma \ref{adap25}, we have $\sup_t \mathbb{E}(T^p_t)<\infty$, which implies that $\sup_{t\geq0}\mathbb{E}(\|\bm P_t\|^p)<\infty$. This completes the proof.
\end{proof}

We then establish the exponential stability of the homogeneous part of the error equation (\ref{adap8}).
\begin{theorem}\label{theo1}
	Consider the distributed FFLS algorithm in Algorithm \ref{algorithm2}.  If
	the forgetting  factor $\alpha$ satisfies $\lambda^{\frac{a^2_{\min}}{32pmh(4h+D_{\mathcal{G}}-1)}}<\alpha<1$ and
	for any $i\in\{1,\cdots,n\}$, $\sup_t\|\bm\varphi_{t,i}\|_{L_{6p}}<\infty$, then under Assumptions \ref{a1} and \ref{a2}, for any $p\geq1$,
	$\{\alpha\bm P_{t+1}\mathscr{A}\bm{ P}^{-1}_{t}\}$ is $L_{p}$-exponentially stable.
\end{theorem}

\begin{proof}
By (\ref{adap6}) and (\ref{adap3}), we have
\ban
\bm P^{-1}_{t+1,i}=\sum^n_{j=1}a_{ij}(\alpha\bm P^{-1}_{t,j}+\bm\varphi_{t,j}\bm\varphi^T_{t,j}).
\ean
Then we can obtain the following equation,
\begin{align*}
tr(\bm P^{-1}_{t+1})&=tr\left(\sum^n_{i=1}\bm P^{-1}_{t+1,i}\right)\\
&= tr\left(\sum^n_{j=1}(\alpha\bm P^{-1}_{t,j}+\bm\varphi_{t,j}\bm\varphi^T_{t,j})\right)\\
&=\alpha tr(\bm P^{-1}_{t})+\sum^n_{j=1}\|\bm\varphi_{t,j}\|^2.
\end{align*}
By  Mikowski inequality, it follows that
\begin{align*}
\|tr(\bm P^{-1}_{t+1})\|_{L_{3p}}
&\leq\alpha \|tr(\bm P^{-1}_{t})\|_{L_{3p}}+O\left(\sum^n_{j=1}\|\bm\varphi_{t,j}\|^2_{L_{6p}}\right)\\
&=\alpha^{t+1}\|tr(\bm P^{-1}_{0})\|_{L_{3p}}+ O\left(\sum^t_{k=0}\alpha^j\right).
\end{align*}
Hence we have
\bna
\sup_{t}\|\bm P^{-1}_{t+1}\|_{L_{3p}}<\infty. \label{adap26}
\ena
By Lemma \ref{yinli2}, we derive that
\ban
&&\Big\|\prod^t_{k=j}\alpha\bm P_{k+1}\mathscr{A}\bm{ P}^{-1}_{k}\Big\|_{L_{p}}\\
&=&\mathbb{E}\left(\Big\|\prod^t_{k=j}\alpha\bm P_{k+1}\mathscr{A}\bm{ P}^{-1}_{k}\Big\|^{p}\right)^{\frac{1}{p}}\\
&=&\mathbb{E}\left(\|\alpha^{t-j+1}\bm P_{t+1}\mathscr{A}^{t-j+1}\bm{P}^{-1}_{j}\|^{p}\right)^{\frac{1}{p}}\\
&\leq&\alpha^{t-j+1}\|\bm P_{t+1}\|_{L_{2p}}\|\bm{P}^{-1}_{j}\|_{L_{2p}}=O(\alpha^{t-j+1}).
\ean
This completes the proof of the theorem.
\end{proof}

Based on Theorem \ref{theo1}, we further establish the tracking error bound of Algorithm \ref{algorithm2} under some conditions on the noises and parameter variation.
\begin{theorem}\label{theo2}
	Consider the model (\ref{model}) and the error equation (\ref{adap8}). Under the conditions of Theorem \ref{theo1},
	if for some $p\geq1$, $\sigma_{3p}\triangleq\sup_t(\|\bm W_{t}\|_{L_{3p}}+\|\Delta\bm\Theta_{t}\|_{L_{3p}})<\infty$, then
	there exists a constant  $c$ such that
	\ban
	\limsup_{t\rightarrow\infty}\|\bm{\widetilde\Theta}_{t}\|_{L_p}\leq c\sigma_{3p}.
	\ean
\end{theorem}

\begin{proof}
For convenience of analysis, 
let the state transition  matrix ${\bm\Psi}(t,k)$ be recursively defined by
\begin{align}
	{\bm\Psi}(t+1,k)=\alpha\bm P_{t+1}\mathscr{A}\bm{ P}^{-1}_{t}{\bm\Psi}(t,k),
	~{\bm\Psi}(k,k)=\bm I_{mn}.
\end{align}
It is clear that ${\bm\Psi}(t+1,k)=\alpha^{t-k+1}\bm P_{t+1}\mathscr{A}^{t-k+1}\bm{P}^{-1}_{k}$.
From the definition of $\bm L_t$ and (\ref{adap6}), we have $\bm{\bar P}^{-1}_{t+1}\bm L_t=\bm\Phi_t$. Then by (\ref{adap8}), we have
\ban
\bm{\widetilde\Theta}_{t+1}=\alpha\bm P_{t+1}\mathscr{A}\bm{ P}^{-1}_{t}\bm{\widetilde\Theta}_{t}-
\bm P_{t+1}\mathscr{A}(\bm\Phi_t\bm W_{t+1}+\bm{\bar P}^{-1}_{t+1}\Delta\bm\Theta_t).
\ean
Hence by H\"{o}lder inequality, we have
\begin{align*}
&\|\bm{\widetilde\Theta}_{t+1}\|_{L_p}\\
=&\Big\|{\bm\Psi}(t+1,0)\bm{\widetilde\Theta}_{0}\\
&-\sum^t_{k=0}{\bm\Psi}(t+1,k+1)(\bm P_{k+1}\mathscr{A}(\bm\Phi_k\bm W_{k+1}+\bm{\bar P}^{-1}_{k+1}\Delta\bm\Theta_k))\Big\|_{L_p}\\
\leq&\|\alpha^{t+1}\bm P_{t+1}\mathscr{A}^{t+1}\bm{ P}^{-1}_{0}\bm{\widetilde\Theta}_{0}\|_{L_p}\\
&+\Big\|\sum^t_{k=0}{\alpha^{t-k}\bm P_{t+1}\mathscr{A}^{t-k+1}(\bm\Phi_k\bm W_{k+1}+\bm{\bar P}^{-1}_{k+1}\Delta\bm\Theta_k)}\Big\|_{L_p}\\
\leq&O(\alpha^{t+1}\|\bm P_{t+1}\|_{L_{2p}})\\
&+\sum^t_{k=0}\alpha^{t-k}\|\bm P_{t+1}\|_{L_{3p}}\|\bm\Phi_k\|_{L_{3p}}
\|\bm W_{k+1}\|_{L_{3p}}\\
&+\sum^t_{k=0}\alpha^{t-k}\|\bm P_{t+1}\|_{L_{3p}}\|\bm{\bar P}^{-1}_{k+1}\|_{L_{3p}}\|\Delta\bm\Theta_k\|_{L_{3p}}.
\end{align*}
Hence by Lemma \ref{yinli2} and (\ref{adap26}), it follows that
\ban
\limsup_{t\rightarrow\infty}\|\bm{\widetilde\Theta}_{t}\|_{L_p}\leq c\sigma_{3p},
\ean
 where $c$  is a positive constant depending on $\alpha$ and the upper
	bounds of $\{\bm P_t\}$,  $\{\bm\Phi_t\}$ and $\{\bm P^{-1}_t\}$.
This completes the proof.
\end{proof}

\begin{remark}
	From the proof of Theorems \ref{theo1} and \ref{theo2},  we can see that if the forgetting factor $\alpha$ is selected to be uncoordinated for different sensors, i.e., we replace $\alpha$ with $\alpha_i$ in Algorithm \ref{algorithm2}, the results of Theorems \ref{theo1} and \ref{theo2}  also hold only if the condition $\lambda^{\frac{a^2_{\min}}{32pmh(4h+D_{\mathcal{G}}-1)}}<\alpha$ is replaced with $\lambda^{\frac{a^2_{\min}}{32pmh(4h+D_{\mathcal{G}}-1)}}<\alpha_{\min}\triangleq\min\{\alpha_1,...,\alpha_n\}$.
\end{remark}

\section{Stability of distributed FFLS algorithm over unreliable directed networks}\label{stability_Markovian}

In Section IV, we have studied the stability of the distributed FFLS algorithm under the fixed undirected graph.
However,
in practical engineering applications,
the information exchange between sensors might not be bidirectional. Moreover, it is often interfered by many uncertain random factors due to the distance, obstacle and interference, which will lead to the interruption or reconstruction of communication links.  Thus, in this section, we model the communication links between sensors as time-varying  random switching directed
communication topologies $\mathcal{G}_{r(t)}=(\mathcal{V}, \mathcal{E}_{r(t)}, \mathcal{A}_{r(t)})$. The switching process is governed by a homogeneous Markov chain $r(t)$ whose states belong to a finite set $\mathbb{S}=\{1,2,...,s\}$, and  the corresponding set of communication topology graph is denoted by $\mathcal{C}=\{\mathcal{G}_1,..., \mathcal{G}_s\}$.
The communication graph is switched just at the instant that the value of $r(t)$ is
changed. Thus, the corresponding  adjacency matrix and
the neighbor set of the sensor $i$ are denoted as   $\mathcal{A}_{r(t)}=[a_{ij,r(t)}]_{1\leq i,j\leq n}$ and
$\mathcal{N}_{i,r(t)}$, respectively.  For the distributed FFLS algorithm over the Markovian  switching directed topologies, we   just modify Step 2 in Algorithm  \ref{algorithm2} as follows:
	\begin{align}
	\bm P^{-1}_{t+1,i}&=\sum_{j\in\mathcal{N}_{i,r(t)}}a_{ji,r(t)}\bm{\bar P}^{-1}_{t+1,j},\label{newada1}\\
	\bm{\hat\theta}_{t+1,i}&=\bm P_{t+1,i}\sum_{j\in \mathcal{N}_{i,r(t)}}a_{ji, r(t)}\bm{\bar P}^{-1}_{t+1,j}\bm{\bar\theta}_{t+1,j}.\label{newada2}
\end{align}
To analyze the stability of algorithm (\ref{adap1}), (\ref{adap2}), (\ref{newada1}), (\ref{newada2}), we introduce the following assumptions:
\begin{assumption}\label{ass1}
 All possible digraphs $\{\mathcal{G}_1,..., \mathcal{G}_s\}$ are balanced and the
union of all those digraphs is strongly connected.
\end{assumption}

 \begin{assumption}\label{ass2}
 	The Markov chain  $\{r_t, t\geq 0\}$ is irreducible and aperiodic
 	with the transition probability matrix $\bm P=[p_{ij}]_{1\leq i,j\leq s}$ where $p_{ij}=\Pr(r_{t+1}=j|r_t=i)$ with $\Pr(\cdot|\cdot)$ being the conditional probability.
 \end{assumption}

According to Markov chain theory (c.f., \cite{Karlin1981}), a discrete-time homogeneous Markov chain with finite states is ergodic if and only if it is irreducible and aperiodic. Hence Assumption \ref{ass2}  means that the $l$-step transition matrix $\bm P^l$ has a limit with identical rows.

In the following, we will analyze the properties of the strongly connected directed graph. For convenience, we denote the $i$-th row, $j$-th column element of the matrix $\bm A$ as $\bm A(i,j)$.
\begin{lemma}\label{lemma5.1}
	Let $\mathcal{G}_k=(\mathcal{V}, \mathcal{E}_k,\mathcal{A}_k), (1\leq k\leq n)$ be $n$ strongly connected graph with $\mathcal{V}=\{1,2,\cdots,n\}$. Then $\mathcal{A}_1\mathcal{A}_2\cdots\mathcal{A}_n$ is a positive matrix, i.e., every element of the matrix $\mathcal{A}_1\mathcal{A}_2\cdots\mathcal{A}_n$ is  positive.
\end{lemma}
\begin{proof}
	We just prove that the graph $\mathcal{G}^n_1$ corresponding to the matrix $\mathcal{A}_1\mathcal{A}_2\cdots\mathcal{A}_n$ is a complete graph.
	 Denote the child node set of the node $i$ in graph $\mathcal{G}_k$ as $\mathcal{O}_{k}(i)$. The corresponding child node set of the node $i$ in graph $\mathcal{G}^n_1$ is denoted by $\mathcal{O}^n_{1}(i)$.
	For any $i\in\mathcal{V}$ and $j\in\mathcal{O}_{1}(i)$,  we have
		 \begin{align} (\mathcal{A}_1\mathcal{A}_2)(i,j)&=\sum^n_{k=1}\mathcal{A}_1(i,k)\mathcal{A}_2(k,j)\nonumber\\
		 	&\geq\mathcal{A}_1(i,j)\mathcal{A}_2(j,j)>0.\label{graph1}
         \end{align}
   Since $\mathcal{G}_2$ is strongly connected, if $\mathcal{O}_{1}(i)\neq\mathcal{V}$, then there exists two nodes $j_1\in\mathcal{V}\backslash\mathcal{O}_{1}(i)$ and $j_2\in\mathcal{O}_{1}(i)$ such that $(j_2,j_1)\in\mathcal{E}_2$, hence
    \begin{align} (\mathcal{A}_1\mathcal{A}_2)(i,j_1)&=\sum^n_{k=1}\mathcal{A}_1(i,k)\mathcal{A}_2(k,j_1)\nonumber\\
   	&\geq\mathcal{A}_1(i,j_2)\mathcal{A}_2(j_2,j_1)>0.\label{graph2}
   \end{align}
By (\ref{graph1}) and (\ref{graph2}), it is clear that $\{j_1\}\cup\mathcal{O}_{1}(i)\subset\mathcal{O}^2_{1}(i)$. Hence for any $j\in \{j_1\}\cup\mathcal{O}_{1}(i)$, we have
	\begin{align} (\mathcal{A}_1\mathcal{A}_2\mathcal{A}_3)(i,j)&=\sum^n_{k=1}(\mathcal{A}_1\mathcal{A}_2)(i,k)\mathcal{A}_3(k,j)\nonumber\\
		&\geq(\mathcal{A}_1\mathcal{A}_2)(i,j)\mathcal{A}_3(j,j)>0.\label{graph3}
	\end{align}	
	 Since $\mathcal{G}_3$ is strongly connected, if $\{j_1\}\cup\mathcal{O}_{1}(i)\neq\mathcal{V}$, then there exists two nodes $j_2\in\mathcal{V}\backslash (\{j_1\}\cup\mathcal{O}_{1}(i))$ and $j_3\in \{j_1\}\cup\mathcal{O}_{1}(i)$ such that $(j_3,j_2)\in\mathcal{E}_3$, hence
	\begin{align} (\mathcal{A}_1\mathcal{A}_2\mathcal{A}_3)(i,j_2)&=\sum^n_{k=1}(\mathcal{A}_1\mathcal{A}_2)(i,k)\mathcal{A}_3(k,j_2)\nonumber\\
		&\geq(\mathcal{A}_1\mathcal{A}_2)(i,j_3)\mathcal{A}_3(j_3,j_2)>0.\label{graph4}
	\end{align}	
By (\ref{graph3}) and (\ref{graph4}), we can see that $\{j_2\}\cup \{j_1\}\cup\mathcal{O}_{1}(i)\subset\mathcal{O}^3_{1}(i)$.		
We repeat the above process until $\mathcal{O}^n_{1}(i)=\mathcal{V}$. The  lemma can be proved by the arbitrariness of the node $i$.
\end{proof}

Compared with the undirected graph case,
the key difference is that the adjacency matrix in this section is an asymmetric and random matrix. Hence we need to deal with the   coupled relationship between random adjacency matrices and  random regression vectors.
By using the above lemma and Markov chain theory, we establish the stability of the algorithm (\ref{adap1}), (\ref{adap2}), (\ref{newada1}), (\ref{newada2}) under Markovian  switching topology.

\begin{theorem}\label{theo3}
Under Assumptions \ref{a2}, \ref{ass1} and \ref{ass2}, if
	for any $i\in\{1,\cdots,n\}$, $\sup_t\|\bm\varphi_{t,i}\|_{L_{6p}}<\infty$ and $\sigma_{3p}\triangleq\sup_t(\|\bm W_{t}\|_{L_{3p}}+\|\Delta\bm\Theta_{t}\|_{L_{3p}})<\infty$ hold,
	then
	there exists a constant  $c'$ such that
	\ban
	\limsup_{t\rightarrow\infty}\|\bm{\widetilde\Theta}_{t}\|_{L_p}\leq c'\sigma_{3p}.
	\ean
\end{theorem}
\begin{proof}
	Following the proof line of Theorem \ref{theo2} in Subsection \ref{results}, it can be seen that we need to prove  equation (\ref{keystep}) holds under the assumptions of the theorem. By Assumption \ref{ass2}, there exists a positive integer $q_0$ such that
	\begin{gather}
	 \Pr(r(t+q_0)=a|r(t)=b)>0\label{gongshi1}
	 \end{gather}
  holds for all $t$ and all states $a,b\in\mathbb{S}$.
	Denote $\Pi^t_k=\mathscr{A}_{r(t)}\mathscr{A}_{r(t-1)}\cdots\mathscr{A}_{r(k)}$. Then the $i$-th row, $j$-th column element of the matrix $\Pi^t_k$ is denoted by $\Pi^t_k(i,j)$. Following Lemmas \ref{lemma1} and \ref{yinli2}, we may abuse some notations $h'=2h+nsq_0$,  $z_t=\lfloor\frac{th'+nsq_0}{h}\rfloor+1$ and
	 \begin{align*}
	 	&b_{t+1}=\\
		&\frac{tr\Big(\sum^{(z_t+1)h}_{k=z_th+1}\sum^n_{j=1}\left(\sum^n_{l=1}\Pi^{k-1}_{th'+1}(j,l)\bm P_{th'+1,l}\right)^2\frac{\bm\varphi_{k,j}\bm\varphi^T_{k,j}}{1+\|\bm\varphi_{k,j}\|^2}\Big)}{n(h'-nsq_0)\left(\alpha^{h'}+\lambda_{\max}\left(\sum^n_{l=1}\bm P_{th'+1,l}\right)\right)tr(\bm P_{th'+1})}.
	\end{align*}
In the following we analyze the term
	$\mathbb{E}(b_{t+1}|\mathscr{F}_{z_th})$.
	By (\ref{gongshi1}), we can see that there exists a positive constant $p_0$ such that for all $t$,
	\begin{align}
		&\Pr\Big(r(t+nsq_0)=s,r(t+(ns-1)q_0)=s-1,\cdots,\nonumber\\
		&~~~~~r(t+((n-1)s+1)q_0)=1;\nonumber\\
		&~~~~~\cdots
		r(t+2sq_0)=s,r(t+(2s-1)q_0)=s-1,\cdots,\nonumber\\
		&~~~~~r(t+(s+1)q_0)=1;\nonumber\\
		&~~~~~r(t+sq_0)=s,r(t+(s-1)q_0)=s-1\cdots,\nonumber\\
		&~~~~~r(t+q_0)=1\Big|\mathscr{F}\Big)\nonumber\\
		=&\sum_{a_0}\Pr\Big(r(t+nsq_0)=s\Big|r(t+(ns-1)q_0)=s-1\Big)\cdots\nonumber\\
		&\Pr\Big(r(t+((n-1)s+1)q_0)=1\Big|r(t+((n-1)s)q_0)=s\Big)\nonumber\\
		&~~~~~\cdots\Pr\Big(r(t+q_0)=1|r(t)=a_0\Big)\Pr\Big(r(t)=a_0\Big|\mathscr{F}\Big)\nonumber\\
		\geq& p_0\sum_{a_0}\Pr\Big(r(t)=a_0\Big|\mathscr{F}\Big)=p_0>0\label{eq48}
	\end{align}
with $\mathscr{F}$ being a $\sigma$-algebra.  By (\ref{eq48}), we know that  the Markov chain $\{r_t,t\geq0\}$ can visit all states  in $\mathbb{S}$ with $n$ times in a positive probability during the time interval $[t+q_0,t+nsq_0]$.
Hence for $k\in[z_th+1,(z_t+1)h)]$, by Assumption \ref{ass1} and Lemma \ref{lemma5.1}, there exists a positive constant $\sigma>0$ such that the following inequality holds,
	\begin{align*}
	&\mathbb{E}\left(\left(\sum^n_{l=1}\Pi^{k-1}_{th'+1}(j,l)\bm P_{th'+1,l}\right)^2\Bigg|\mathscr{F}_{k}\right)\\
	=&\mathbb{E}\Big(\Big(\sum_{u\in\mathcal{V}}\sum_{v\in\mathcal{V}}\Pi^{k-1}_{th'+1}(j,u)\Pi^{k-1}_{th'+1}(j,v)\bm P_{th'+1,u}\bm P_{th'+1,v}\Big)\Big|\mathscr{F}_{k}\Big)\\
	=&\sum_{u\in\mathcal{V}}\sum_{v\in\mathcal{V}}\Big(\mathbb{E}\Big(\Pi^{k-1}_{th'+1}(j,u)\Pi^{k-1}_{th'+1}(j,v)\Big)\Big|\mathscr{F}_{k}\Big)\bm P_{th'+1,u}\bm P_{th'+1,v}\\
	\geq&\sigma\sum_{u\in\mathcal{V}}\sum_{v\in\mathcal{V}}\bm P_{th'+1,u}\bm P_{th'+1,v}=\sigma\left(\sum^n_{l=1}\bm P_{th'+1,l}\right)^2.
	\end{align*}
  By $\mathscr{F}_{z_th}\subset\mathscr{F}_{k}$ and $\bm\varphi_{k,j}\in\mathscr{F}_{k}$, we conclude that
\begin{align}
	&\mathbb{E}\left(\left(\sum^n_{l=1}\Pi^{k-1}_{th'+1}(j,l)\bm P_{th'+1,l}\right)^2\frac{\bm\varphi_{k,j}\bm\varphi^T_{k,j}}{1+\|\bm\varphi_{k,j}\|^2}\Bigg|\mathscr{F}_{z_th}\right)\nonumber\\
	=&\mathbb{E}\Bigg(\Bigg(\mathbb{E}\left(\sum^n_{l=1}\Pi^{k-1}_{th'+1}(j,l)\bm P_{th'+1,l}\right)^2\Bigg|\mathscr{F}_{k}\Bigg)\nonumber\\
	&\cdot\frac{\bm\varphi_{k,j}\bm\varphi^T_{k,j}}{1+\|\bm\varphi_{k,j}\|^2}\Bigg|\mathscr{F}_{z_th}\Bigg)\nonumber\\
	\geq&\sigma\mathbb{E}\Bigg(\left(\sum^n_{l=1}\bm P_{th'+1,l}\right)^2
	\frac{\bm\varphi_{k,j}\bm\varphi^T_{k,j}}{1+\|\bm\varphi_{k,j}\|^2}\Bigg|\mathscr{F}_{z_th}\Bigg).
\end{align}	
From the above analysis, we can obtain the following inequality
	\begin{align*}
		&\mathbb{E}(b_{t+1}|\mathscr{F}_{z_th})\geq\\
		& \frac{tr\Big(\sum^{(z_t+1)h}_{k=z_th+1}\sum^n_{j=1}\sigma\mathbb{E}\Big(\left(\sum^n_{l=1}\bm P_{th'+1,l}\right)^2
			\frac{\bm\varphi_{k,j}\bm\varphi^T_{k,j}}{1+\|\bm\varphi_{k,j}\|^2}\Big|\mathscr{F}_{z_th}\Big)
			\Big)}{n(h'-nsq_0)\Big(\alpha^{h'}+\lambda_{\max}\Big(\sum^n_{l=1}\bm P_{th'+1,l}\Big)\Big)tr(\bm P_{th'+1})}\\
		&=\frac{\sigma tr\left[\left(\sum^n_{l=1}\bm P_{th'+1,l}\right)^2\bm H_{z_t}
			\right]}{n(h'-nsq_0)\left(\alpha^{h'}
			+\lambda_{\max}\left(\sum^n_{l=1}\bm P_{th'+1,l}\right)\right)tr(\bm P_{th'+1})}.
	\end{align*}
	The rest part of the proof can be obtained by following  the proofs of Lemma \ref{yinli2}, Theorems \ref{theo1} and \ref{theo2} just replacing the notation $D_{\mathcal{G}}$ with $nsq_0$. This completes the proof of Theorem \ref{theo3}.
\end{proof}
\begin{remark}
From Theorem \ref{theo3}, (also Theorems \ref{theo1} and \ref{theo2}), we see that
	our results are obtained without using the independency
	or stationarity assumptions on the regression signals, which
	makes it possible to apply the distributed FFLS algorithm to practical feedback systems.
	\end{remark}

\section{Concluding Remarks}\label{section_conclusion}
This paper proposed a distributed FFLS algorithm to collaboratively track an unknown time-varying parameter  by minimizing a local loss function with a forgetting factor. By introducing a spatio-temporal cooperative excitation condition, we established
the  stability of the proposed distributed FFLS algorithm for fixed undirected graph case. 
Then, the theoretical results were generalized to the case of Markovian switching directed graphs. 
The cooperative excitation condition
revealed  that the sensors can collaboratively  accomplish the tracking task even though any individual sensor cannot.
We note that our theoretical results
are established without using independence or stationarity conditions of the regression vectors. Thus, a relevant research topic is how to combine the distributed adaptive estimation with the distributed control.
How to establish the stability analysis of the distributed  algorithms for
more complex cases such as considering quantization effect or time-delay in
communication channels is another interesting research topic.

\bibliographystyle{IEEEtran}
\bibliography{my_bib_FFLS}

\end{document}